\documentclass[prx,aps,twocolumn,superscriptaddress,groupedaddress,nofootinbib]{revtex4-1}
\usepackage{amsmath,amsfonts,amsthm,bbm,graphicx,color,enumerate,hyperref,xcolor,tikz,pgfplots,ulem,subfigure}
\usepackage{algpseudocode}
\usepackage{algorithm}

\hypersetup{
	colorlinks=true,  
	linkcolor=blue,   
	citecolor=blue,   
	urlcolor=blue     
}

\def\Tr{{\rm{tr}}}

\def\>{\rangle}
\def\<{\langle}

\def\hc{^{\dagger}}

\renewcommand{\emph}{\textit}

\newtheorem{theorem}{Theorem}

\newtheorem{lemma}[theorem]{Lemma}

\newtheorem{definition}[theorem]{Definition}

\newtheorem{corollary}[theorem]{Corollary}

\newtheorem{innercustomthm}{Theorem}
\newenvironment{customthm}[1]{\renewcommand\theinnercustomthm{#1}\innercustomthm}{\endinnercustomthm}

\begin{document}
\title{ Efficient recovery of variational quantum algorithms landscapes using classical signal processing}
\author{Enrico Fontana$^{1,2,3}$}
\author{Ivan Rungger$^3$}
\author{Ross Duncan$^{1,2,4}$}
\author{Cristina C\^{i}rstoiu$^{1}$}

\affiliation{$^1$Quantinuum, Terrington House, 13-15 Hills Road, Cambridge CB2 1NL, UK}
\affiliation{$^2$Department of Computer and Information Sciences, University of Strathclyde, 26 Richmond Street, Glasgow G1 1XH, UK}
\affiliation{$^3$National Physical Laboratory, Hampton Road, Teddington TW11 0LW, UK}
\affiliation{$^4$Department of Physics and Astronomy, UCL, Gower Street, London, WC1E 6BT, UK}

\begin{abstract}
We employ spectral analysis and compressed sensing to identify settings where a variational algorithm's cost function can be recovered purely classically or with minimal quantum computer access. We present theoretical and numerical evidence supporting the viability of sparse recovery techniques. To demonstrate this approach, we use basis pursuit denoising to efficiently recover simulated Quantum Approximate Optimization Algorithm (QAOA) instances of large system size from very few samples. Our results indicate that sparse recovery can enable a more efficient use and distribution of quantum resources in the optimisation of variational algorithms. 
\end{abstract}

\maketitle

\section{Introduction}
As the available quantum hardware starts to reach scales where quantum advantage may be possible \cite{arute2019quantum}, it becomes crucial to understand under what conditions the heuristic algorithms employed in the NISQ regime can beat classical methods.  
A central object in a variational quantum algorithm (VQA) is the cost function, which may be seen as a parameterised expectation value encoding the problem of interest. For example, the variational quantum eigensolver~\cite{peruzzo2014variational} (VQE) finds the ground state energy by minimising the Hamiltonian's expectation value. 

Performance guarantees for VQAs face a series of hurdles. If the cost function can be efficiently classically estimated then there is no need for a quantum computer.    The complexity of this mean value task was explored in \cite{bravyi2021classical}, which found a $\mathcal{O}(N)$ classical algorithm for constant-depth 2D circuits on $N$ qubits. Furthermore, the parametrised circuits also need to be expressive enough to obtain the desired accuracy \cite{holmes2022connecting}. To avoid classical simulability and achieve expressibiliy, this requires scaling the number of parameters and circuit depth with the problem size -- which in turn places stronger constraints on the cost function optimisation.

Global features of the landscape \cite{arrasmith2021equivalence, holmes2022connecting} can impact the training of parametrised circuits. The clasical optimisation of cost functions used in VQAs  has been shown to be computationally hard (NP-hard) \cite{bittel2021training}. Phenomena like barren plateaus - where the cost function has an exponentially decaying variance with increasing system size -  appear as a result of concentration in random circuits \cite{mcclean2018barren, uvarov2021barren}, non-locality of the cost function and circuit structure \cite{cerezo2021cost, holmes2022connecting, larocca2021diagnosing, pesah2021absence, marrero2021entanglement} or noise \cite{wang2021noise}. These results led to various strategies seeking to avoid the barren plateaus \cite{wiersema2020exploring, volkoff2021large, patti2021entanglement, sack2022avoiding, larocca2021theory, haug2021capacity}. 
Recent works analysed cost functions in terms of their Fourier expansions \cite{schuld2021effect, fontana2022spectral, vidal2018calculus, nakanishi2020sequential, parrish2019jacobi, ostaszewski2021structure}. In \cite{koczor2022quantum} an efficient method to approximately determine the cost function in the vicinity of a point was used to accelerate gradient descent.

Motivated by our previous work \cite{fontana2022spectral} that relates circuit structures with Fourier transform of the cost function to filter the effects of noise,  we investigate when the cost function can be recovered either fully classically or with minimal quantum queries. To this aim, sparsity in the Fourier basis plays a significant role and enables compressed sensing methods  \cite{donoho2006compressed}. These are used in a wide range of application including tomography \cite{gross2010quantum, flammia2012quantum}, eigenvalue estimation \cite{somma2019quantum}. This work  introduces compressed sensing in the context of near term quantum algorithms.

Firstly, we show that under appropriate conditions the cost function can be deduced from the circuit structure alone. When this fails, we show that for circuits with correlated parameters one can still use limited access to a quantum computer to fully recover the cost function. Our notion of \textit{recoverable cost function} implies efficient scaling in both sampling and computational complexity with system size. 

Secondly, we find numerical evidence for large VQAs whose cost functions are \textit{sparse} in the Fourier basis, which means most coefficients are (approximately) zero.
We propose the use of \textit{Basis Pursuit Denoising} (BPDN), an efficient compressed sensing method which is well suited for noisy settings \cite{chen2001atomic}. The advantage lies in reconstructing the cost function with a reduced number of samples, which can be taken randomly. 

Finally, we showcase an application of the method to QAOA \cite{farhi2014quantum}, by replacing the quantum-classical optimisation loop with a randomised sampling step followed by classical post-processing. The resulting algorithm is sample-efficient, and separates the quantum and classical runtimes, eliminating the need for low-latency solutions \cite{karalekas2020quantum} and allowing for asynchronous, distributed applications.

\section{Exact cost function recovery}
We consider VQAs with parametrised circuits $U(\boldsymbol{\theta}) $ on $N$ qubits, that encode the problem of interest in the cost function given by the expectation value of observable $O$:
\begin{equation}
	C(\boldsymbol \theta) := \<O\>_{\psi(\boldsymbol \theta)} = \text{Tr}[O U(\boldsymbol \theta) \psi_0 U^{\dagger}(\boldsymbol \theta)],
\end{equation}
where $\psi_0$ is some initial state and $\psi(\boldsymbol \theta) = U(\boldsymbol \theta) \rho_0 U^{\dagger}(\boldsymbol \theta)$.
This is a real valued function on $M$ independent (continuous) parameters $\boldsymbol{\theta}\in [0,2\pi]^M$ and it has been shown several times in the literature \cite{schuld2019evaluating, vidal2018calculus, nakanishi2020sequential, parrish2019jacobi, ostaszewski2021structure, schuld2021effect, koczor2022quantum} that its Fourier representation 
\begin{equation}
	\<O\>_{\psi(\boldsymbol{\theta})} = \sum_{\bf{k}\in \boldsymbol \Lambda}  \hat{c}_{\bf{k}}(O) e^{i \boldsymbol{\omega}_{\bf{k}} \cdot \boldsymbol{\theta}}
	\label{eq:thm1}
\end{equation}
has a finite number of coefficients indexed by a set $\Lambda$ with bounded frequencies in each dimensions, thus making it a \textit{trigonometric polynomial}. Furthermore, in previous work \cite{fontana2022spectral} we relate Fourier coefficients $\hat{c}_{\bf{k}} (O)$ and frequency spectra to the structure of variational circuits (see Appendix~\ref{app:fourier}). Throughout this work we consider $N$ to be the number of qubits, $n$ to be the cardinality of $\Lambda$ (i.e total number of non-zero frequencies), $m$ the sample size and $M$ to be the number of parametrised single-qubit rotations in $U(\boldsymbol{\theta})$ with $d\leq M$ independent parameters.

\subsection{Calculating cost functions from variational circuits}
Given the connection between Fourier coefficients and circuit structure, we highlight situations when the closed form of the cost function can be determined classically. We consider the following circuit classes that we call \emph{Clifford variational circuits} and take the form:
\begin{equation}
	U(\boldsymbol \theta) = C_M e^{-i P_M \theta_M/2} C_{M-1} \cdots C_1 e^{-i P_1 \theta_1/2} C_{0},
\end{equation} 
where the $M$ independently parametrised unitaries are generated by Pauli operators $P_i$ and where each unparametrised unitary $C_i$ is Clifford.
The closed form of the corresponding cost function $C(\boldsymbol{\theta})$ has at most $3^M$ terms, with each coefficient computable in polynomial time in $N$. 

\begin{theorem}
	For a VQA with Clifford variational circuits of depth $S$ (total number of gates), let $|\psi_0\>$ be the input state with stabiliser weight $\omega_\rho$ and let the measured observable $O$ decompose into $\omega_O$ Pauli operators. 
	If $M \in O(\log N)$ and $w_{\rho}, w_O, S \in O(\text{\normalfont{poly}}(N))$, then the cost function $C(\boldsymbol \theta)$ can be computed in closed form by a polynomial time in $N$ classical algorithm. 
	\label{thm:2}
\end{theorem}

Well-known ans\"atze that take the above form include the Hardware-efficient Ansatz (HEA) \cite{kandala2017hardware} and the Unitary Coupled Cluster (UCC) \cite{romero2018strategies}.
Typically $w_O \in O(\text{\normalfont{poly}}(N))$ as measuring the observable needs to be efficient on a quantum computer.
In the case of HEA the number of parameters scales as $\text{\normalfont{poly}}(N)$ even at constant depth, and therefore Theorem~\ref{thm:2} does not apply directly.
On the other hand, while the UCC typically has a reference state that is a computational basis state \cite{mcclean2016theory} and therefore has $w_S = 1$, in its single and double excitations implementation the number of parameters scales polynomially as $M \in O(N^4)$, assuming half-filling \cite{peruzzo2014variational}. However, algorithms have been proposed that achieve good accuracy with a much lower number of parameters \cite{ryabinkin2018qubit, grimsley2019adaptive, tang2021qubit, lee2018generalized}. For example, ADAPT-VQE applied to LiH attains chemical accuracy with just 8 parameters \cite{grimsley2019adaptive}. Our result emphasizes that with increasing problem size the number of parameters needs to scale appropriately for shallower UCC-type circuits to avoid classical simulability.

\subsection{Cost function recovery with quantum queries}
\label{sec:access}
We analyse general classes of VQAs, whose cost functions can be reconstructed from a limited number of samples on a set of parameter values. 
\begin{definition}[Efficiently recoverable cost function]
\normalfont
A  VQA has an \textit{efficiently recoverable} cost function if there exists a $O({\text{poly}}(N))$ classical algorithm that determines the cost function up to a given additive accuracy from a set of samples $\{C(\boldsymbol{\theta}_i) \;|\; \boldsymbol{\theta}_i \in \boldsymbol \Theta \}$ taken on a parameter set $\boldsymbol \Theta\subset [0,2\pi]^d$ with cardinality $m \in O(\text{poly}(N))$.
\end{definition}
Whilst the definition includes any sampling scheme and classical post-processing, we focus on Fourier methods. The structure of parametrised circuits alone determines the maximal allowed frequency and finite support $\Lambda$ (see Appendix \ref{app:fourier} and  \cite{fontana2022spectral}). However, this does not directly determine which of the allowed frequencies appear in $C(\boldsymbol{\theta})$ nor the value of Fourier coefficients $\hat{c}_{\boldsymbol{k}}$.

If the cost function is sampled on a sufficiently fine grid, we are guaranteed to recover $\hat{c}_{\mathbf{k}} $ and $\omega_{\bf{k}}$ via the higher-dimensional Discrete Fourier Transform (DFT).  The $d$-dimensional Fast Fourier Transform (FFT)  will achieve this in $O(m \, \text{log} \, m)$ time for circuits that contain $d$-independent parameters.

\begin{theorem}[Correlated parameter cost functions are efficiently recoverable]
\label{thm:3}
Consider VQAs with circuits composed of $M \in O(\text{\normalfont{poly}}(N))$ parameterised single-qubit rotation gates and any number of unparameterised gates.
Then the cost function is efficiently recoverable via DFT if the number of independent parameters $d$ does not vary with $N$. 
\end{theorem}
The theorem is proved in Appendix \ref{app:thm:3}.
Theorem \ref{thm:3} would apply, for example, to heavily correlated ans\"atze like QAOA \cite{farhi2014quantum} or HVA \cite{wiersema2020exploring} where the depth is limited to a sufficiently low constant. Then the resulting cost function will be recoverable.
Notice that these conditions already cover interesting systems that are not known to be directly classically simulable. For example QAOA circuits with $p$ layers where the parameters are set to just two superparameters: $\beta_i = \beta,\, \gamma_i = \gamma, \, \forall \, i \le p$. Such correlation strategies have been explored to reduce the dimension of parameter space \cite{moussa2022unsupervised}. In this case we only need $p$ to scale polynomially in $N$ for the QAOA cost function to be efficiently recoverable via this method, although such circuits avoid current classical simulators for sufficiently large $N$.

\section{Sparse cost function recovery}

\subsection{Basis Pursuit}
We have so far considered settings in which cost functions of VQAs can either be simulated via completely classical methods, or reconstructed via hybrid methods whereby the landscape is sampled on the quantum computer.
Both approaches require exponential resources in the number of parameterised rotation gates. In particular, the total number of samples $m$ required to resolve the band-limited frequency spectrum will scale exponentially with the dimension of $\boldsymbol{\theta}$. 

However, efficient recovery is indeed possible for \textit{any} cost function, provided that it is \textit{sparse} in the Fourier basis -- namely that the number $s$ of (significant) frequencies is small. This is enabled by a host of classical signal processing techniques that has been developed to reconstruct sparse signals. For example, the $s$-sparse FFT can be implemented in $O(s \log^{O(1)} n)$ time \cite{kapralov2016sparse}.
We will focus on Basis Pursuit (BP) \cite{chen2001atomic}, which solves an L1 optimisation problem to find a sparse representation. This choice is motivated by its practical effectiveness, but other techniques such as Matching Pursuit \cite{mallat1993matching} and Orthogonal Matching Pursuit \cite{pati1993orthogonal} are also available. For a comparison between these techniques in the context of trigonometric polynomial recovery see \cite{kunis2008random}.
Since quantum cost functions are naturally probabilistic, we work with the extension of Basis Pursuit to signals corrupted by noise \cite{chen2001atomic}, known as Basis Pursuit Denoising (BPDN), which seeks:
\begin{equation}
	\text{min}_{\, \boldsymbol{\hat{c}}}  \ \frac{1}{2} \|\Phi\boldsymbol{\hat c} - \boldsymbol{C}\|^2_2 + \lambda \|\boldsymbol{\hat c}\|_1,
\end{equation}
where $\hat{\bf{c}}$ is the vector of all $n$ Fourier coefficients, $\Phi$ is the matrix obtained by selecting $m$ rows of the $n\times n$ DFT matrix corresponding to the sampling points $\{\boldsymbol{\theta}_i\}$, and $\boldsymbol{C} : = [C(\boldsymbol{\theta_1}),..., C(\boldsymbol{\theta_m}))]$ is the vector of observations of the cost function. Details about BP and BPDN can be found in Appendix \ref{app:bp}.

\subsection{Narrow gorge landscapes are unrecoverable}
The \textit{narrow gorge} is a phenomenon that is known to occur in highly parameterised and expressive VQAs like the HEA \cite{cerezo2021cost}. The name refers to the variational cost function being probabilistically concentrated around a value that is significantly different from the optimum. The $N$-qubit cost function $C(\boldsymbol{\theta})$ of a VQA class has a narrow gorge if there is a value $\boldsymbol{\theta}^{*}$ for which $C(\boldsymbol{\theta^{*}}) \in \Omega(1/\text{poly}(N))$ and has exponential vanishing variance $\sigma_{\boldsymbol{\theta}}[C(\boldsymbol{\theta})] \in O(b^{-N})$ for $b>1$ and any $N$.

The presence of narrow gorges therefore makes local optimisation from random starting conditions challenging, which is compounded by their close association to barren plateaus \cite{arrasmith2021equivalence}.  We show (in Appendix \ref{app:finite}) that narrow gorges also prevent efficient cost function recovery.
\begin{theorem}[Narrow gorge landscapes are unrecoverable]
	\label{thm:gorge_unrec_finite}
	Consider VQAs with a narrow gorge and periodic cost function.
	Then, for sufficiently large $N$ with probability $\ge 1 - 1/{\normalfont \text{poly}}(N)$ the cost function is not recoverable by BPDN if it is sampled on a quantum computer at $m \in O(\text{\normalfont poly}(N))$ randomly chosen points with at most polynomially many shots (measurements) per sample. 
\end{theorem}
The unrecoverability of narrow gorge landscapes can be seen as a manifestation of the uncertainty principle of Fourier analysis \cite{donoho1989uncertainty}. A narrow gorge is indeed a concentration of a function in a small region of parameter space. The fact that cost functions presenting a narrow gorge are unrecoverable via Basis Pursuit, even without noise (see Theorem \ref{thm:gorge_unrec} in Appendix \ref{app:finite}), in turn implies that they cannot be too sparse in the Fourier basis. To be more precise, the theorems proved here and \cite[Corollary 2.2]{rauhut2007random} together imply that sparsity must obey $s > \frac{m}{C \log(n/\epsilon)}$ for any approximation error$\epsilon$ and constant $C$. Therefore, since $n \sim c^N$, and the theorems are valid for any $m \in \text{poly}(N)$, it must mean that the sparsity is superpolynomial in $N$.

\subsection{Do provably recoverable quantum cost functions exist?}
At this point one is left to wonder whether there can exist any VQAs that are provably recoverable by sparse recovery methods. Based on Theorem \ref{thm:gorge_unrec_finite}, we can exclude those algorithms that are known to present narrow gorges and barren plateaus. Overall this points us away from overparameterised ans\"atze like the HEA, and towards VQAs with correlated parameters such as QAOA and HVA, which are known to not exhibit barren plateaus, under certain conditions \cite{larocca2021diagnosing}.
In the next sections we provide numerical evidence that, in some cases, the cost functions of these VQAs are sparse in Fourier space, and hence efficiently recoverable.

\section{Numerical experiments}
\begin{figure}
\centering
\includegraphics[width=0.5\textwidth]{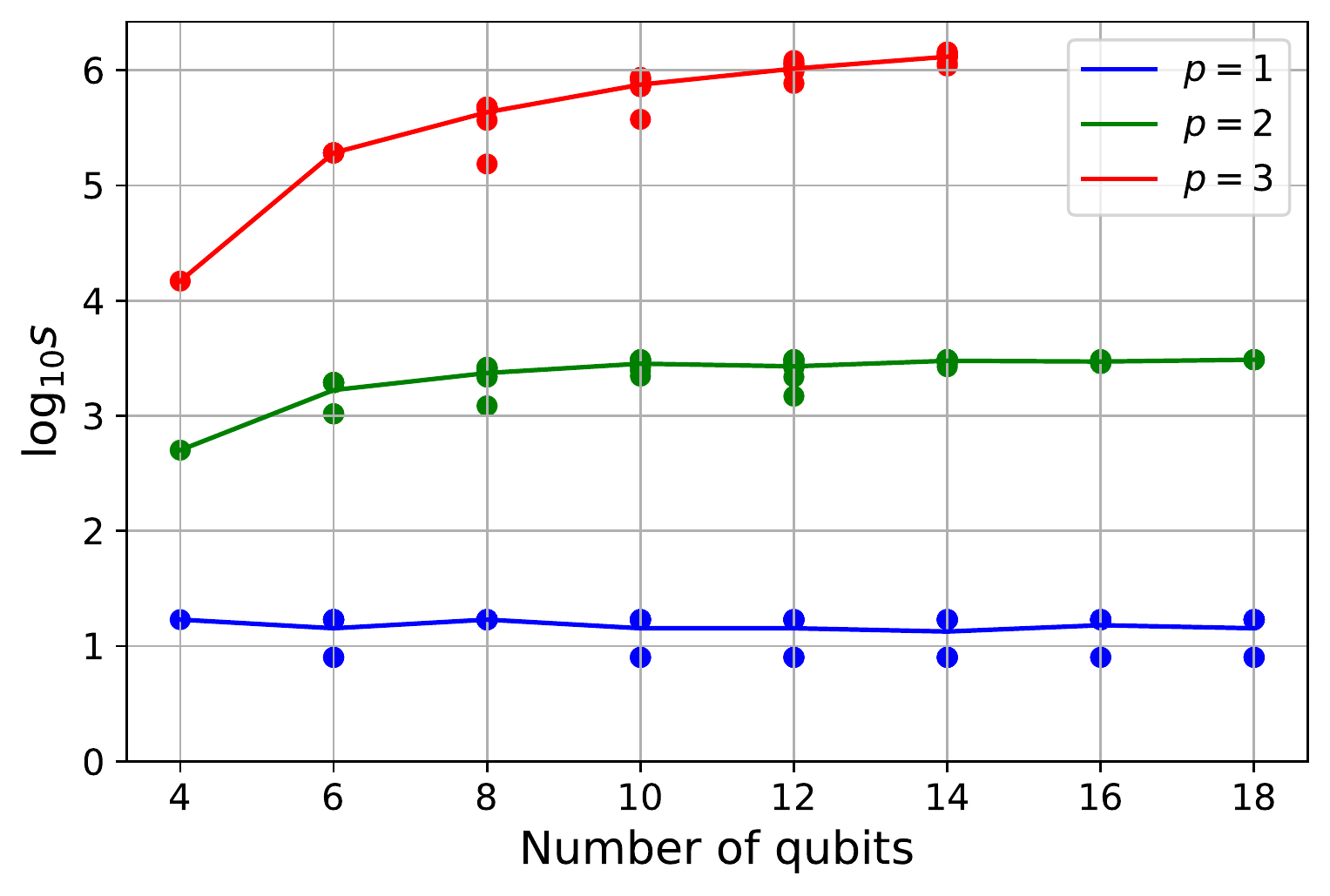}
\caption{Scaling of sparsity $s$ of Fourier coefficients for QAOA on MaxCut for 10 random 3-regular graphs, for different number of qubits $N$ and layers $p$. The circuits have $d =2p$ independent parameters. The circles are the outcomes for single graphs and the line shows the average trend.}
\label{fig:sparsity-plot}
\end{figure}

\subsection{Measuring sparsity}
We focus on QAOA \cite{farhi2014quantum}, for the MaxCut combinatorial problem on random 3-regular graphs. Our approach to measure sparsity consists of sampling the exact cost function on a (uniform) grid covering the parameter space at a resolution determined by the maximum support (as indicated by Theorem 2 in \cite{fontana2022spectral}), and applying multidimensional FFT on the resulting data, thus returning all the Fourier coefficients. We make use of the Qulacs statevector simulator \cite{suzuki2021qulacs}.
The results are shown in Figure \ref{fig:sparsity-plot}. The sparsity is seen to scale favourably with the number of qubits $N$. On the other hand, the scaling with the number of layers $p$ is less favourable, displaying an exponential trend.
Since sampling from the entire grid is required to calculate sparsity there exists an unavoidable limitation to size of the algorithm that can be investigated. As such the analysis was limited to $p=3$, $N=14$. This highlights the fact that exact reconstruction quickly becomes impractical even at limited depth.

\subsection{Recovering the QAOA landscape}
The numerical results for sparsity encouraged us to attempt recovering the QAOA cost function, focusing on large systems with shallow depth.
We performed the reconstruction with a commonly used approach to solving BPDN problems, the Fast Iterative Shrinkage-Thresholding Algorithm (FISTA) \cite{beck2009fast}. 
We found it beneficial to follow the algorithm with an additional L2-only minimisation step, which refines the coefficients on the identified support only. 
The complete algorithm for recovery is outlined as \textsc{Recover} in Algorithm~\ref{alg:1} in Appendix \ref{app:algo}.

In Figure \ref{fig:mse_trend} we show the trend of recovery accuracy for $p=2$ and number of qubits between 16 and 28, measured by the mean-squared error (MSE) on an out-of-sample (not used for BPDN) set of 100 random landscape points, as a fraction of the mean squared cost function value. The number of samples used for BPDN ranges from 1000 to 4000. While FISTA already gives a low reconstruction error, the \textsc{Recover} algorithm yields more accurate solutions, up to one order of magnitude lower MSE. The performance appears not to scale significantly with system size.
\begin{figure}
\centering
\includegraphics[width=0.5\textwidth]{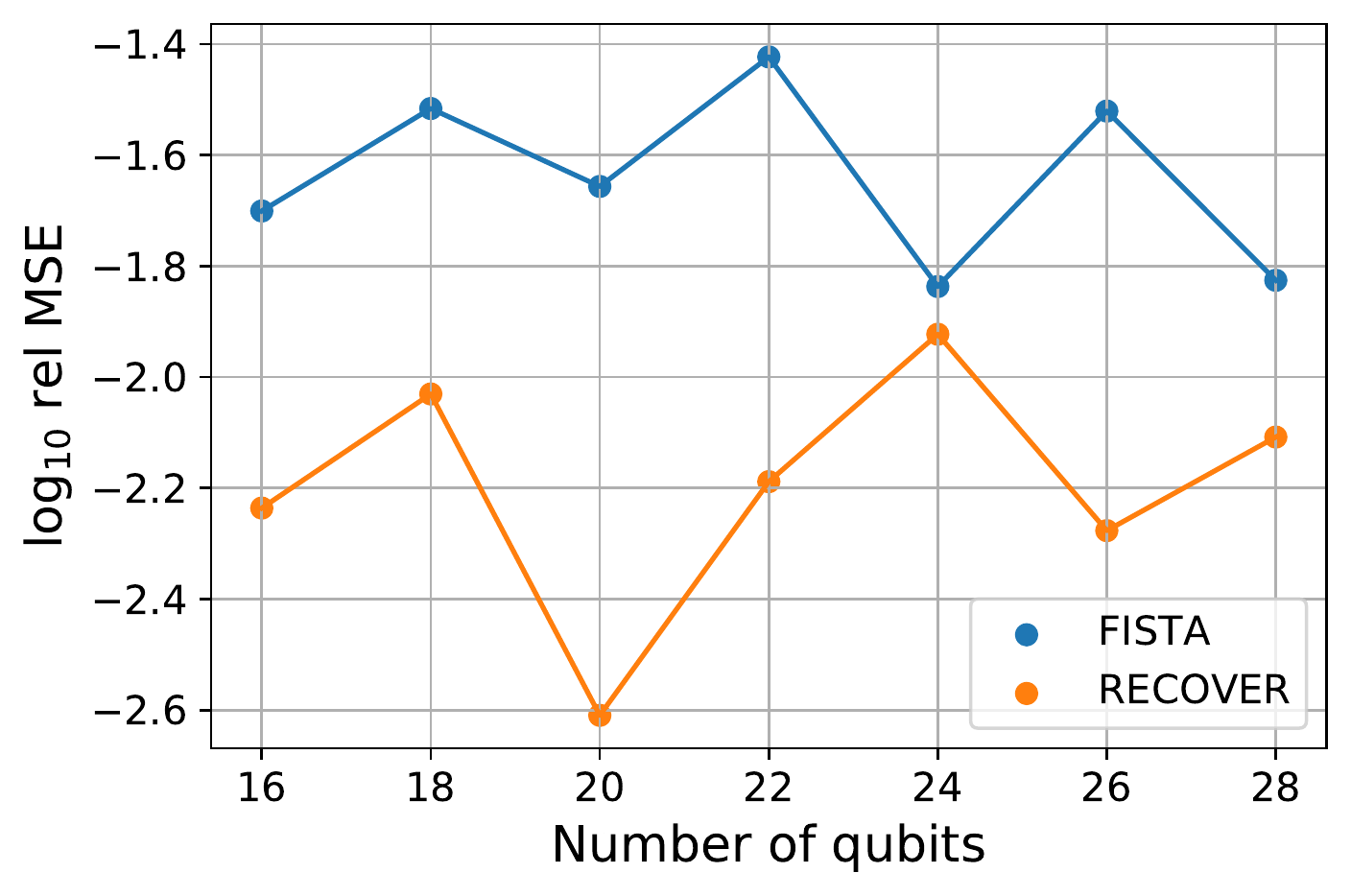}
\caption{Trend of out-of-sample relative mean-squared error of reconstructed cost functions with number of qubits, $p=2$ QAOA for MaxCut on a random 3-regular graph. Comparison between FISTA and \textsc{Recover} in Algorithm~\ref{alg:1}.}
\label{fig:mse_trend}
\end{figure}
\subsection{Enhancing QAOA optimisation}
\begin{figure*}
\centering
\includegraphics[width=\textwidth]{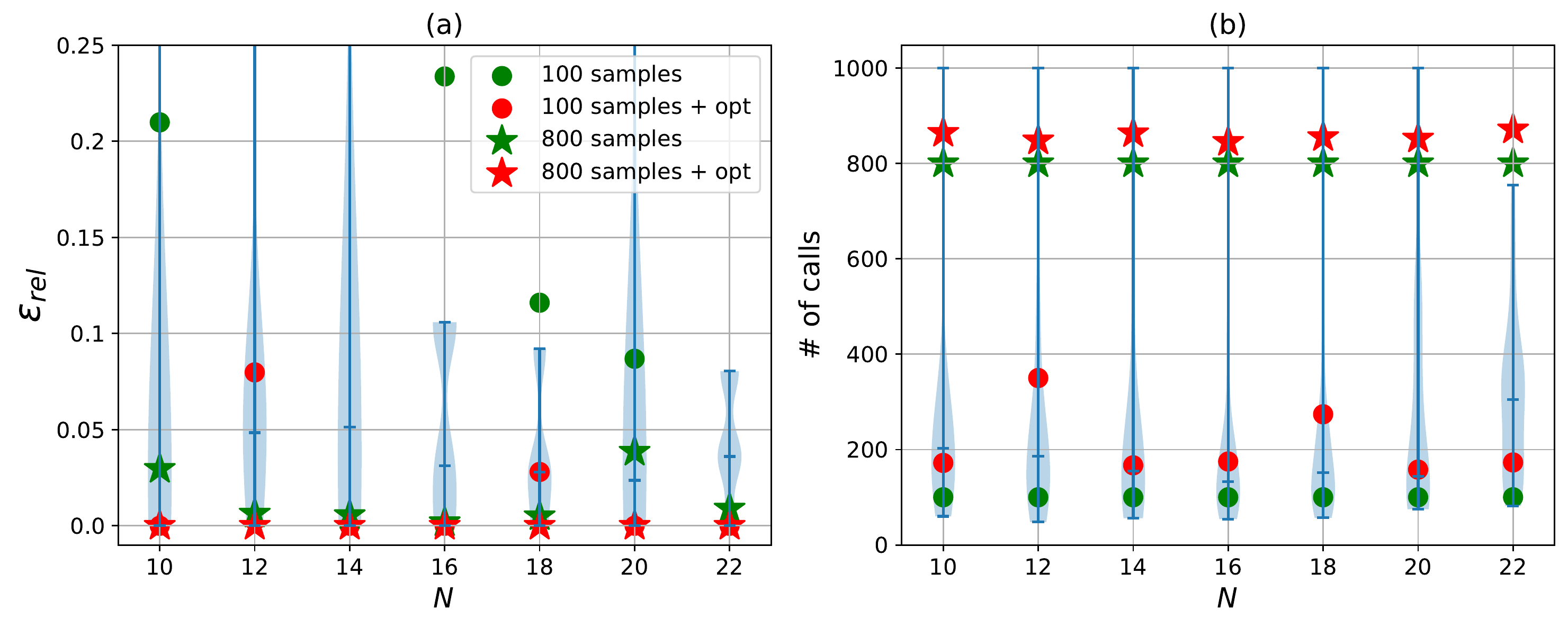}
\caption{Quality of optimum parameters found on recovered landscapes (green markers), compared to 100 runs of randomly initialised GD (blue violin plots, faint line is median value).
Figure (a) shows the relative error compared to the true minimum value, (b) the total number of calls to the quantum computer.
Two separate reconstructions, with 100 (circles) and 800 (stars) samples were performed. The red markers show the result after a further round of gradient descent.}
\label{fig:recon_vs_gd}
\end{figure*}
Next we examine a potential application of cost function recovery.
The method's strengths lie in its low sampling requirement to reach a good level of accuracy and the fact that the samples can be taken completely at random.
It is natural therefore to consider applying BPDN to enhance gradient-based VQA optimisation, as the latter procedure typically suffers from a large number of samples that depends on a continuous feedback between classical and quantum.

We perform a simulation of $p=2$ QAOA, this time performing gradient descent (GD) on the reconstructed landscape. As the \textsc{Recover} algorithm returns a closed form for the cost function, the optimisation is performed on a classical device. We then compare the relative error in true cost function value between the optimal parameters from the reconstructed cost function and the exact global optimum. This is shown in Figure \ref{fig:recon_vs_gd} (a). We also consider the total number of calls to the quantum computer, shown in Figure \ref{fig:recon_vs_gd} (b) and we show the outcome of applying a further round of GD on the true cost function, using the reconstructed minima as starting points. All reconstructions are performed twice, using 100 (circles) and 800 samples (stars) respectively, and compare to distributions of 100 runs of GD with random initialisation (blue violin plots).

The cost functions recovered from 100 samples generally do not yield the exact location of the optimum, however they are good starting points for GD, generally achieving a better accuracy than random initialisation while requiring a similar overall number of calls to the quantum computer.
On the other hand, 800 samples typically yield a high accuracy optimum without any further GD. Although this is more than the median number of quantum computer calls with random GD, it should be stressed that the samples are taken in a completely asynchronous manner while random GD requires a classical-quantum loop.

\section{Discussion}
Our methods for cost function reconstruction have several practical applications.
First, spectral and compressed sensing methods can be used to determine regimes (i.e circuit structure, parameter regions, applications) where cost functions of VQAs can be approximated via classical means. Sparse recovery in combination with quantum simulators may be used to (approximately) recreate entire VQA cost function. Since sparse reconstruction is partially resilient to noise in the input samples, one may use simulators that return approximate expectation values of quantum states, such as those based on tensor networks \cite{zhou2020limits} or neural networks \cite{jonsson2018neural, medvidovic2021classical}.
The second application is distributing the workload in a VQA across multiple distinct machines running parallel quantum computations, with the goal of realising speed-ups in computational runtime and more efficient allocation of quantum resources. Basis Pursuit is particularly well suited for this goal as the samples can be taken asynchronously at random, and therefore the sampling could be performed on different quantum computers at different times.

The main novelty in our approach is the focus on sparsity in the Fourier basis. Clearly, if sparsity could be rigorously proven for classes of VQAs with increasing problem sizes it would have important real-world implications. 
The results also raise the possibility that cost functions of some VQAs with large system size may be efficiently simulable by classical means. In the case of QAOA, closed form solutions exist for $p=1$ \cite{hadfield2018quantum}; furthermore, concentration results suggest that the QAOA landscape at large $N$ may be well approximated by smaller systems \cite{brandao2018fixed}.
This latter fact may be a reason why the sparsity in Figure \ref{fig:sparsity-plot} appears to be stable as $N$ grows.

Finally, the idea of predicting the output of quantum computations with a classical model trained on quantum data is shared by the techniques in quantum machine learning (QML), and carries important implications for quantum advantage \cite{huang2021power}. The present work was developed independently from \cite{schreiber2022classical} which introduced recently the concept of classical surrogates for QML models. Their method is similar to the one  in Section \ref{sec:access}, however it is specific to a QML setting.

\section{Conclusion}
Motivated by understanding the regimes in which VQAs can offer practical quantum advantage, we explore whether spectral analysis and compressed sensing can be used to efficiently construct the cost function of VQAs.
For this we employ recent works that decompose quantum channels into Fourier process modes \cite{cirstoiu2017global, fontana2022spectral, koczor2022quantum}.

We found a positive answer for practically-relevant circuits with a low number of independent parameters but not directly limited in depth or system size.  Next, we applied methods from classical signal processing to evaluate sparse recovery of the cost function from a few samples. We presented numerical evidence that low depth QAOA cost functions are sparse in the Fourier domain, and demonstrated efficient recovery. The results may carry important implications for the optimisation of VQAs and their advantage over classical methods.

\section{Acknowledgements}
We would like to thank Raul Garcia Patron and Michał Stechły for useful discussions and Steven Herbert, Konstantinos Meichanetzidis and Marcello Benedetti for their helpful comments and feedback on this manuscript. EF and IR acknowledge the support of the UK government department for Business, Energy and Industrial Strategy through the UK national quantum technologies programme. EF acknowledges the support of an industrial CASE (iCASE) studentship, funded by the UK Engineering and Physical Sciences Research Council (grant EP/T517665/1), in collaboration with the university of Strathclyde, the National Physical Laboratory, and Quantinuum.

\bibliography{main}

\onecolumngrid
\appendix
\vspace{0.5in}

\begin{center}
  { \bf \MakeUppercase{Appendix}}
\end{center}

\section{Notation and conventions}
We use $N$ to refer to the number of qubits, $M$ for the number of parameters, $n$ for the total number of frequency modes, $m$ for the number of sampled points from a landscape. When the effective dimension of parameter space is different from the number of parameters, for examples because some parameters are correlated, we indicate the former by $d$.\\
We define a trigonometric polynomial of maximum frequency $n$ as:
\begin{equation}
f(x) = \sum_{k=-n}^n \hat{f}_k e^{i x k}
\end{equation}
The Fourier coefficients are recovered via the Discrete Fourier Transform (DFT):
\begin{equation}
\hat{f}_k = \frac{1}{2\pi} \int_{0}^{2\pi} f(x) e^{-i x k} dx
\end{equation}

\section{Fourier spectrum of quantum cost functions}
\label{app:fourier}
We start by reviewing the central result in \cite{fontana2022spectral}:
\begin{theorem}[Spectral characteristics of VQA landscapes \cite{fontana2022spectral}]
	\label{thm:1}
	Consider a parameterised quantum state $\rho(\boldsymbol{\theta})= U(\boldsymbol{\theta})\rho_0 U(\boldsymbol{\theta})\hc$ acting on $N$ qubits with $M$ independent parameters:
	\begin{equation}
		U(\boldsymbol{\theta}) = V_0 U(\theta_1)V_1 ... U(\theta_M)V_M 	
	\end{equation}
	where $U(\theta_j) = e^{iH_j \theta_j}$ are one-parameter unitaries generated by Hermitian operators $H_j$. The expectation value of any Hermitian operator $O$ with respect to $\rho(\boldsymbol{\theta})$ is a \textit{generalised $M$-variate trigonometric polynomial} of bounded degree:
	\begin{equation}
		\<O\>_{\rho(\boldsymbol{\theta})} = \sum_{\bf{k}\in \boldsymbol \Lambda}  \hat{c}_{\bf{k}}(O) e^{i \boldsymbol{\omega}_{\bf{k}} \cdot \boldsymbol{\theta}}
		\label{eq:thm1}
	\end{equation}
	where $\hat{c}_{\bf{k}} ^{*} = c_{-\bf{k}}$ and the lattice $\boldsymbol{\Lambda} \subset \mathbb{Z}_3^{r_1}\times ... \times \mathbb{Z}_3^{r_M}$   where each $r_j$ is the rank of the Walsh-Hadamard transform of the eigenvalue vector of $H_j$.\footnote{The $2^l$-dim Walsh-Hadamard transform matrix can be defined as $[W_l]_{ij} := \frac{1}{2^{l/2}}(-1)^{i\cdot j}$ where the $i\cdot j$ indicates the bitwise dot product between the binary representations of the indices.} Furthermore, the frequencies vector $\boldsymbol{\omega}_{\bf{k}} = ( (\omega_{\bf{k}})_1, ...., (\omega_{\bf{k}})_M)$ ranges over a discrete set with bounded degree $\sup_{\bf{k} \in \boldsymbol \Lambda}|(\omega_{\bf{k}})_j |\leq 2\|H_j\|_{\infty}$.
\end{theorem}
Note the shorthand $\mathbb{Z}_{r} := \{-r, -r+1, ..., r-1, r\}$.
The theorem is similar in spirit to results proven in \cite{schuld2019evaluating, vidal2018calculus, schuld2021effect}, however takes a different approach and is in effect more general as it applies to gates generated by any Hermitian $H_j$. This however comes at the expense of applicability. To make the findings directly relatable to typical quantum computing scenarios, we restrict it to circuits where the one-parameter unitaries are independent, canonically parameterised and generated by Pauli operators: $U_i(\theta) = e^{-i P_i \theta/2}$, where $P_i$ is in the Pauli group and $\theta \in (-\pi, \pi]$.
In this setting, the Theorem takes a simpler form:
\begin{corollary}[Theorem~\ref{thm:1} revisited]
	Consider a parameterised quantum state $\rho(\boldsymbol{\theta})$ generated by the action on a state $\rho_0$ of a variational circuit as described above, with $M$ independent parameters. Then $\<O\>_{\rho(\boldsymbol{\theta})}$ is a generalised $M$-variate trigonometric polynomial of bounded degree. The frequencies are supported on the lattice $\boldsymbol \Lambda = \mathbb{Z}_3^{M}$.
\end{corollary}

Up to now we have assumed independently parameterised rotation gates. The result can however be extended in a straight-forward way to any variational ans\"atze of interest currently include different parameterised gates that are controlled by the same parameter:
\begin{corollary}[Correlated gates]
\label{rem:1}
If any one set of $M_c$ rotation gates as described above are correlated by setting them to the same parameter, the frequency support of $\<O\>_{\rho(\boldsymbol{\theta})}$ in terms of the reduced set of independent parameters will become
	\begin{equation}
		\boldsymbol \Lambda = \mathbb{Z}_3^{M - M_c} \times \mathbb{Z}_{2M_c + 1}.
	\end{equation}
\end{corollary}
In such situations, we will refer to the independent parameters as \textit{superparameters} to distinguish them from the single-gate parameters.
The two corollaries in various forms have been reported in \cite{vidal2018calculus, nakanishi2020sequential, parrish2019jacobi, ostaszewski2021structure, schuld2021effect, koczor2022quantum}.

Overall, the take-home messages of these results are: 1) VQAs have a finite, discrete frequency support, and 2) the maximum such theoretical support be derived from the form of the parameterised unitary in a straight-forward manner. 

\section{Proof of Theorem \ref{thm:2}}
\label{app:thm:2}
\begin{customthm}{1}
	Consider a VQA with an ansatz given by an $N$ qubit circuit that consists of $M$ parametrised unitaries generated by Pauli operators $P_i$ and where all $S$ unparametrised gates are Clifford:
	\begin{equation}
		U(\boldsymbol \theta) = C_M e^{-i P_M \theta_M/2} C_{M-1} \cdots C_1 e^{-i P_1 \theta_1/2} C_{0}.
	\end{equation} We refer to $U(\boldsymbol{\theta})$ of this form as a \emph{Clifford variational circuit}. Let $|\psi_0\>$ be the input state with stabiliser weight $\omega_\rho$ and the measured observable $O$ decompose into $\omega_O$  Pauli operators. 
	Then if $M \in O(\log N)$ and $w_{\rho}, w_O, S \in O(\text{\normalfont{poly}}(N))$, the closed form of the cost function $C(\boldsymbol \theta)$ can be determined by a polynomial time in $N$ classical algorithm. 
\end{customthm}
\begin{proof}
Let us write
\begin{equation}
	|\psi_0\> = \sum_i^{w_{\rho}} s_i |\psi_i\>, \;\;\; O = \sum_i^{w_O} r_i P_i ,
\end{equation}
where $|\psi_i\>$ are stabilizer states and $P_{i}$ are $N$-qubit Pauli operators.
By linearity, the cost function can be decomposed into a sum over stabilisers and Pauli operators:
\begin{equation}
	C(\boldsymbol \theta) = \sum_{ijl} s_i s^*_j r_l \; \Tr\left[U^\dagger(\boldsymbol \theta) P_l U(\boldsymbol \theta) |\psi_i\>\<\psi_j|\right]
\end{equation}
Using the results in \cite{fontana2022spectral}, we can decompose the superoperator $\mathcal{U}^\dagger (\boldsymbol \theta) := U^\dagger (\boldsymbol \theta) \cdot U(\boldsymbol \theta)$ into Clifford process modes. Therefore, working in the Heisenberg picture, we can write $U^\dagger(\boldsymbol \theta) P_l U(\boldsymbol \theta) = \sum_{k} \phi_{k} (\boldsymbol \theta) P_{lk}$, where $P_{lk}$ is another Pauli operator and $\phi_{k} (\boldsymbol \theta)$ is a trigonometric monomial. This sum has at most $3^M$ terms.\\
Combining with the above one gets:
\begin{equation}
	C(\boldsymbol \theta) = \sum_{ijkl} s_i s^*_j r_l \phi_{k} (\boldsymbol \theta) \; \Tr\left[P_{lk} |\psi_i\>\<\psi_j|\right]
\end{equation}
The sum has at most $3^M w^2_\rho w_O$ terms. Each Pauli operator $P_{lk}$ can be determined in linear time in $S$ by applying successive Clifford channels to $P_l$, each step scaling linearly with $N$. The expectation values can also be determined in time linear in $N$ by virtue of the states being stabiliser. Therefore, the entire sum can be determined in polynomial time in $N$ as long as $M \in O(\log N)$, $w_{\rho}, w_O, S \in O(\text{\normalfont{poly}}(N))$.
\end{proof}

\section{Proof of Theorem \ref{thm:3}}
\label{app:thm:3}
\begin{customthm}{3}[Correlated parameter cost functions are efficiently recoverable]
Consider a VQA where the ansatz is composed of $M \in O(\text{\normalfont{poly}}(N))$ parameterised single-qubit rotation gates and any number of unparameterised gates. Let $d$ be the number of independent parameters with the $M$ rotation angles partitioned into $d$ sets, each containing multiples of the same parameter.  Then the algorithm's cost function is efficiently recoverable via DFT if $d$ does not vary with $N$. 
\end{customthm}
\begin{proof}
Assume, without loss of generality, that single-qubit rotation gates are canonically parameterised, and that the parameter within each set are correlated by setting them equal to the same superparameter. Denote the number of parameters within set $i$ by $M_i$, such that $\sum_{i=1}^d M_i = M$. By Corollary~\ref{rem:1}, the maximum size of the Fourier support is therefore
\begin{equation}
	|\Lambda| = \prod_{i=1}^d (2M_i + 1) \le \left(2\frac{M}{d} + 1\right)^d
\end{equation}
Since $M \in O(\text{\normalfont poly}(N))$ and $d$ is a constant, $|\Lambda| \in O(\text{\normalfont poly}(N))$. Therefore, the entire cost function can be recovered by sampling from a polynomial sized grid and performing a DFT.
\end{proof}

\section{Review of Basis Pursuit}
\label{app:bp}
Basis Pursuit is the following algorithm:
\begin{definition}[Basis Pursuit \cite{chen2001atomic}]
	\normalfont
	Consider a function $f(x) \in \mathbb{R}$ sampled over a set of points $\Omega = \{x_i\} \subset \mathbb{R}$. Basis Pursuit attempts to reconstruct the function by a trial function $g(x)$ which is the solution of the following convex linear program:
	\begin{equation}
		\label{eq:linprog}
		\text{minimize } \|\hat{g}\|_1 := \sum_k |\hat{g}_k| \text{ s.t. } g(x_i) = f(x_i) \; \forall \; x_i \in \Omega
	\end{equation}
	where $\hat{g}_k$ are the coefficients of $g(x)$ when expanded in some basis of functions.
\end{definition}
We are mostly concerned about the case when $f$ is a periodic function of bounded frequency (a trigonometric polynomial) and the basis is the discrete Fourier basis, which therefore includes finitely many terms $n$ (the non-periodic case is considered in Appendix~\ref{app:nonuniform}). However, we usually have for the sampling set $|\Omega| = m \ll n$, i.e. the problem is underdetermined and may admit multiple solutions. The power of Basis Pursuit emerges in the case the signal is sparse in the Fourier basis.
The effectiveness in this case has been rigorously shown in the theorem of Rauhut \cite{rauhut2007random}, which is based on the results by Candes \textit{et al.} \cite{candes2006robust}. The following formulation is that of \cite{kunis2008random}.
\begin{theorem}[Rauhut \cite{rauhut2007random}]
	\label{thm:rauhut}
	Consider a trigonometric polynomial that has $s$ nonzero Fourier coefficients out of $n$ possible ones, and is otherwise unknown. Then for some $\epsilon > 0$ the polynomial can be reconstructed from $m$ random samples with probability $\ge 1 - \epsilon$ via Basis Pursuit, if
	\begin{equation}
		m \ge Cs\log(n/\epsilon)
		\label{eq:sampling_req}
	\end{equation}
	where C is an absolute constant.
\end{theorem}

Basis Pursuit has been extended to noisy signals with a known amount of random noise. Remarkably, recovery of the noiseless signal is possible, hence the name \textit{Basis Pursuit Denoising}.
The following formulation comes from \cite{donoho2005stable}:
\begin{definition}[Basis Pursuit Denoising (BPDN) \cite{chen2001atomic}]
	\normalfont
	Consider a noisy signal $f(x) \in \mathbb{R}$ sampled over a set $\Omega = \{x_i\} \subset \mathbb{R}$. The signal corresponds to a noiseless signal $f_0(x)$ corrupted by random Gaussian noise $e(x)$: $f(x) = f_0(x) + e(x)$. Basis Pursuit Denoising attempts to reconstruct the entire noiseless signal by solving the following convex linear program:
	\begin{equation}
		\text{minimize } \|\hat{g}\|_1 \text{ s.t. } \frac{1}{|\Omega|} \sum_{x_i \in \Omega} (g(x_i) - f(x_i))^2 \le \epsilon
	\end{equation}
	where $\epsilon := \frac{1}{|\Omega|} \sum_{x \in \Omega} e^2(x) $ defines the tolerance to noise.
\end{definition}
Note that BPDN is equivalent to Basis Pursuit whenever $\epsilon = 0$.
There exist equivalent results on the effectiveness of Basis Pursuit Denoising for the case of sparse signals when sufficient samples are taken, see e.g. \cite{candes2006stable}.

For trigonometric polynomials, BPDN can be formulated in vector form as (from \cite{chen2001atomic}):
\begin{equation}
		\text{minimize } \frac{1}{2} \|\Phi\boldsymbol{\hat g} - \boldsymbol{f}\|^2_2 + \lambda \|\boldsymbol{\hat g}\|_1
\end{equation}
where $\boldsymbol{\hat g}$ is the vector of all $n$ Fourier coefficients, which is varied until convergence. $\Phi$ is the matrix obtained by selecting $m$ rows of the $n\times n$ DFT matrix corresponding to the sampling points $\{x_1, \cdots, x_m\}$, and $\boldsymbol{f}$ is the vector of observations of the function $f$: $\boldsymbol{f} := [f(x_1), \cdots, f(x_m)]$. $\lambda$ is a tunable parameter that enforces sparsity in the solution: a larger value will lead to a sparser solution, which however may not be the optimal one.
This is also known as a LASSO problem \cite{tibshirani1996regression} and allows for the use of iterative solvers.

So far we have assumed 1D signals, however the results extend trivially to any dimensionalty $D$ by substituting $n \rightarrow n^D$. This is because there exists an isomorphism between a $n$-frequency DFT in $D$ dimensions and a $n^D$-frequency DFT in 1 dimensions. With this substitution, from Theorem \ref{thm:rauhut} the number of samples $m$ necessary to have a high probability of success would scale linearly with the dimensionality $D$, assuming $s$ does not vary with $D$.

\section{Reconstruction algorithm}
\label{app:algo}
Our specific implementation is an gradient-descent-based solver called the Fast Iterative Shrinkage-Thresholding Algorithm (FISTA)~\cite{beck2009fast}, specific to LASSO optimisation problems. 
If successful, the output of the FISTA algorithm returns the Fourier support alongside the corresponding coefficients. However, due to the soft thresholding step the latter might not be the optimal ones. Therefore, we have found it beneficial to follow the algorithm with an additional L2-only minimisation step, which refines the coefficients on the identified support only. In our case we achieve this with a conjugate gradient algorithm~\cite{fletcher1964function}, but simple gradient descent also succeeds.\\
The complete algorithm for reconstruction is outlined as \textsc{Recover} in Algorithm~\ref{alg:1}. 
An example of the effectiveness of the algorithm for reconstructing a QAOA landscape can be seen in Figure \ref{fig:mse_28}. The extra conjugate gradient step can be seen to visibly improve the accuracy.

\begin{figure*}
\centering
\includegraphics[width=0.9\textwidth]{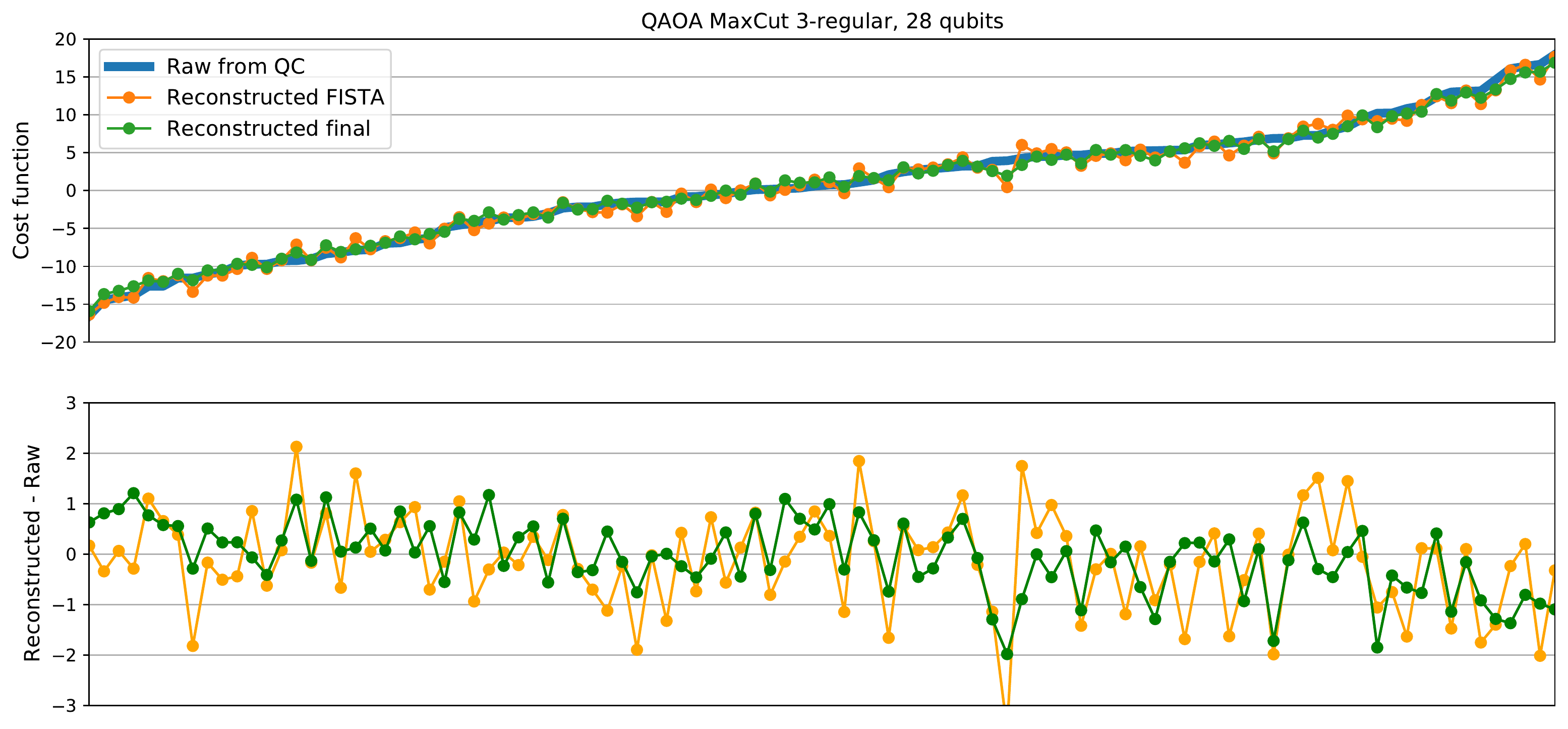}
\caption{Out-of-sample accuracy of reconstruction for a $N=28$, $p=2$ QAOA, performed using $m=4000$ random samples of the landscape.
Top: cost function value of 100 random points, for the true landscape (blue), the reconstructed landscape with FISTA only (orange), the reconstructed landscape after FISTA and gradient descent (green).
Bottom: difference in cost function values between true and reconstructed landscapes.}
\label{fig:mse_28}
\end{figure*}

\begin{algorithm}[H]
\caption{Method for recovering Fourier coefficients of VQA}
\label{alg:1}
\begin{algorithmic}[1]
\Procedure{FISTA}{$\boldsymbol \theta, \boldsymbol C, \alpha, \lambda, n$}
	\State $\boldsymbol{\hat c} \gets \boldsymbol 0$
	\State $\Delta \boldsymbol{\hat c} \gets \boldsymbol 0$

	\For{$i \gets 1$ \textbf{to} $n$}
		\State $\boldsymbol{g} \gets $ \textsc{Grad}($\boldsymbol \theta, \boldsymbol C, \boldsymbol{\hat c}$) \Comment{Gradient of L2-only norm cost}
		\State $\boldsymbol{\hat c'} \gets \boldsymbol{\hat c} - \alpha \boldsymbol{g} + \frac{i - 2}{i + 1} \Delta \boldsymbol{\hat c}$ \Comment{Accelerated GD step}
		\State $\boldsymbol{\hat c''} \gets \text{sign}(\boldsymbol{\hat c'}) \max(|\boldsymbol{\hat c'}| - \alpha\lambda\boldsymbol 1, \boldsymbol 0)$  \Comment{Soft thresholding}
		\State $\Delta \boldsymbol{\hat c} \gets \boldsymbol{\hat c''} - \boldsymbol{\hat c}$
		\State $\boldsymbol{\hat c} \gets \boldsymbol{\hat c''}$
	\EndFor

	\State \textbf{return} $\boldsymbol{\hat c}$ \Comment{Converged Fourier coeffs}
\EndProcedure

\Procedure{Recover}{$m, \alpha_{FISTA}, \lambda, n_{FISTA}, \alpha_{GD}, n_{GD}$}
		\State $\boldsymbol \theta \gets$ \textsc{RandomSampleGrid}(m) \Comment{Get $m$ random points on grid in parameter space}
		\State $\boldsymbol C \gets$ \textsc{QCEval}($\boldsymbol \theta$) \Comment{Query QC to generate samples}
		\State $\boldsymbol{\hat c} \gets$ \textsc{FISTA}($\boldsymbol \theta, \boldsymbol C, \alpha_{FISTA}, \lambda, n_{FISTA}$) \Comment{Run FISTA to get sparse coefficient vector}
		\State $\boldsymbol{\hat c'} \gets$ \textsc{GD}($\boldsymbol{\hat c}, \boldsymbol \theta, \boldsymbol C, \alpha_{GD}, n_{GD}$)  \Comment{Gradient descent over sparse support}
		\State \textbf{return} $\boldsymbol{\hat c'}$ \Comment{Converged Fourier coeffs}
\EndProcedure

\end{algorithmic}
\end{algorithm}

\subsubsection{Hyperparameter tuning}
Throughout our experiments we set $\alpha_{FISTA} = 0.1, n_{FISTA} = 3000, \alpha_{GD} = 1, n_{GD} = 40$.
The parameter $\lambda$ determines the threshold for the Fourier coefficients, and thus must be determined based on their expected magnitude. Heuristically we find that a good choice is $\lambda \sim \frac{1}{m}\|\boldsymbol C\|^2_2$.
The number of samples $m$ should be based on the intended accuracy of the algorithm as well as the cost of sampling, either on a simulator or on a quantum computer, which may vary greatly based on the specific setting. Prior knowledge on the expected sparsity may also be taken into account when choosing sampling size.
In our implementation, we determine $m$ dynamically, by setting aside a small (size 100) set of datapoints and comparing the out-of-sample (OOS) error between the zero vector and the reconstructed cost function after the algorithm has completed. If the error is greater than a threshold, $m$ is increased accordingly and the algorithm is re-run. Notice that this strategy can be made sample-efficient by reusing previous samples from the quantum computer in the new iteration.

\section{Narrow gorges are unrecoverable}
\label{app:finite}

\begin{definition}[Narrow gorge \cite{cerezo2021cost}]
	\label{def:narrow_gorge}
	\normalfont
	Consider a family of VQAs indexed by their system size $N$. The set of cost functions $\{C^{(N)}(\boldsymbol{\theta})\}_N$ is said to present a \textit{narrow gorge} if the following two conditions are met:
	\begin{itemize}
	 	\item For any $N$, there exists a $\boldsymbol{\theta}^*$ such that
			\begin{equation}
				\label{eq:gorge}
				|C^{(N)}(\boldsymbol{\theta}^*)| \in \Omega(1/{\normalfont \text{poly}}(N)) .
			\end{equation}
		\item For any $N$, there exists a constant $b > 1$ such that
			\begin{equation}
				\label{eq:var}
				\underset{\boldsymbol{\theta}}{\sigma}[C^{(N)}(\boldsymbol{\theta})] \in O(b^{-N}) .
			\end{equation}
	 \end{itemize}
\end{definition}

\subsection{Proof of noiseless result}

We first need the following two lemmas:
\begin{lemma}
	\label{lem:1norm}
	With probability $\ge 1 - \frac{1}{t^2}$, the 1-norm of a vector $\boldsymbol{y}$ of $m$ iid random variables obeying $\mathbb{E}(y_i) = 0, \text{Var}(y_i) = \sigma^2$ obeys
	\begin{equation}
		\|\boldsymbol{y}\|_1 \le (m + t\sqrt{m})\sigma.
	\end{equation}
\end{lemma}
\begin{proof}
Since the $y_i$ are iid we have trivially:
\begin{align}
	\mathbb{E}(\|\boldsymbol y\|_1) &= m \mathbb{E}(|y_i|)\\
	\text{Var}(\|\boldsymbol y\|_1) &= m \text{Var}(|y_i|) = m\mathbb{E}(y_i^2) - m\mathbb{E}^2(|y_i|) = m\sigma^2 - m\mathbb{E}^2(|y_i|)
\end{align}
Now use the inequality
\begin{equation}
	\frac{1 + u - (u - 1)^2}{2} \le \sqrt{u} \le \frac{1 + u}{2}
\end{equation}
Substituting $u = \frac{y_i^2}{\sigma^2}$ such that $\mathbb{E}(u) = 1$ and taking expectation values, we can write
\begin{equation}
	\mathbb{E}(|y_i|) = \sigma(1 - \Delta)
\end{equation}
where $0 \le \Delta \le \frac{1}{2\sigma^4}\text{Var}(y_i^2)$.
With this expression we can now write
\begin{align}
	\mathbb{E}(\|\boldsymbol y\|_1) &= m\sigma(1 - \Delta), \\
	\text{Var}(\|\boldsymbol y\|_1) &= m\sigma^2 - m\sigma^2(1 - \Delta)^2.
\end{align}
Finally we use Chebyshev's inequality to assert that
\begin{equation}
	\textbf{Pr}\left(\left|\|\boldsymbol y\|_1 - \mathbb{E}(\|\boldsymbol y\|_1) \right| \ge t\sqrt{\text{Var}(\|\boldsymbol y\|_1)}\right) \le \frac{1}{t^2}.
\end{equation}
Therefore with probability higher than $1-\frac{1}{t^2}$ the following inequality occurs:
\begin{align}
	\| \boldsymbol{y}\|_1 \le \left|\|\boldsymbol y\|_1 - \mathbb{E}(\|\boldsymbol y\|_1) \right|  + \mathbb{E}(\|\boldsymbol y\|_1 ) \le  t\sqrt{\text{Var}(\|\boldsymbol y\|_1)} +  \mathbb{E}(\|\boldsymbol y\|_1  = m\sigma(1-\Delta) + t\sqrt{m\sigma^2 - m\sigma^2(1-\Delta)^2} ,
\end{align}
where the first bound comes from the triangle inequality. Maximising over all  values of $\Delta$ implies that the RHS is upper bounded by $(m + t\sqrt{m})\sigma$.

\end{proof}

\begin{lemma}
\label{lem:recon_fails}
Take a signal $\boldsymbol{y} \in \mathbb{C}^n$ to be reconstructed from its samples $\boldsymbol{y}|_{\Omega}$ over a set $\Omega$, via Basis Pursuit in Fourier space.
Then whenever
\begin{equation}
	\label{eq:cond}
	\frac{\|\boldsymbol{y}|_{\Omega}\|_1 }{\|\boldsymbol{y}\|_{2}} < \frac{1}{\sqrt{n}},
\end{equation}
there exists a vector $\boldsymbol{z} \neq \boldsymbol{y}$ such that $\boldsymbol{z}|_{\Omega} = \boldsymbol{y}|_{\Omega}$ and $\|\boldsymbol{\hat{z}}\|_1 < \|\boldsymbol{\hat{y}}\|_1$, and therefore the reconstruction fails.
\end{lemma}

\begin{proof}
Consider the vector $\boldsymbol{z}$ defined by $\boldsymbol{z}|_{\Omega} = \boldsymbol{y}|_{\Omega}$, $\boldsymbol{z}|_{\Omega^c} = 0$. We claim that whenever the condition \ref{eq:cond} above holds,
\begin{equation}
 	\|\boldsymbol{\hat{z}}\|_1 \le \sqrt{n} \|\boldsymbol{y}|_{\Omega}\|_1 < \|\boldsymbol{y}\|_{2} \le \|\boldsymbol{\hat{y}}\|_1 .
\end{equation}
The first and third inequalities follow from successive application of classic relations between norms and Parseval's identity.
The second inequality in the chain comes from imposing condition \ref{eq:cond}.

We assumed 1D, however the same result can be proven for $d$ dimensions, with the only difference being that now $n$ is the size of the full $d$-dimensional space.
\end{proof}

With this we can prove the following theorem:
\begin{theorem}[Narrow gorge landscapes are unrecoverable (noiseless)]
	\label{thm:gorge_unrec}
	If a family of VQAs indexed by their system size $N$ presents in its cost function a narrow gorge by Definition \ref{def:narrow_gorge} with decay constant $b$, then for sufficiently large $N$ with high probability $\ge 1 - 1/{\normalfont \text{poly}}(N)$ it is not recoverable by Basis Pursuit using the linear program \ref{eq:linprog} for all $b$ if the sampling rate $m \in O(\normalfont \text{poly}(N))$ and the support $n$ scales subexponentially with $N$. The statement is true for $n \sim c^N$ provided that $b > \sqrt{c}$. 
\end{theorem}
\begin{proof}
Assume that we sample the landscape from a regular grid that is fine enough to resolve the correct Fourier coefficients, if enough samples were to be taken, i.e. at least at the Nyquist rate. Indicate the resulting vector on the entire grid as $\boldsymbol C_{all}$.
The narrow gorge condition (\ref{eq:gorge}) directly implies that $\|\boldsymbol{C_{all}}\|_{2} \in \Omega(1/\text{\normalfont poly}(N))$. This can be seen by applying Parseval's relation twice: $\|\boldsymbol C_{all}\|_2 =  \|\boldsymbol{\hat c}\|_2 = \left(\int |C(\boldsymbol{\theta})|^2 d\boldsymbol\theta \right)^{1/2}$.

Consistently with the notation used in the text, we indicate the vector of values sampled from a random set of size $m$ as $\boldsymbol C$. From Lemma \ref{lem:1norm} one can easily see that if one chooses $t, m \in \text{\normalfont poly}(N)$, and one applies the narrow gorge condition (\ref{eq:var}): $\sigma \in O(b^{-N})$, then with high probability $\ge 1 - \frac{1}{\text{\normalfont poly}(N)}$ one has $\|\boldsymbol{C}\|_1\in O(b^{-N})$. 

Altogether, 
\begin{equation}
	\frac{\|\boldsymbol{C}\|_1 }{\|\boldsymbol{C}_{all}\|_2} \in O(b^{-N}) .
\end{equation}
Since by our assumptions $n \in O(c^{N})$ and $\sqrt{c} < b$, for sufficiently large system sizes $N$ Lemma \ref{lem:recon_fails} holds and the minimiser converges to an incorrect solution.
\end{proof}

\subsection{Proof of result with finite sampling noise (Theorem \ref{thm:gorge_unrec_finite})}

Similarly to before, we make use of a lemma on the 2-norm of a vector of random variables:
\begin{lemma}
	\label{lem:2norm}
	With probability $\ge 1 - \frac{1}{t^2}$, the 2-norm of a vector $\boldsymbol{y}$ of $m$ iid random variables obeying $\mathbb{E}(y_i) = 0, \text{Var}(y_i) = \sigma^2$ obeys
	\begin{equation}
		\|\boldsymbol{y}\|_2 \le (1 + t)\sqrt{m}\sigma.
	\end{equation}
\end{lemma}
\begin{proof}
We start from
\begin{equation}
	\mathbb{E}(\|\boldsymbol y\|^2_2) = m \mathbb{E}(y_i^2) = m\sigma^2
\end{equation}
Now use the same inequality as before, this time taking $u = \frac{\|\boldsymbol y\|^2_2}{m\sigma^2}$, we get:
\begin{equation}
	\mathbb{E}(\|\boldsymbol y\|_2) = \sqrt{m}\sigma(1 - \Delta)
\end{equation}
Now with $0 \le \Delta \le \frac{1}{2m\sigma^4}\text{Var}(y_i^2)$.
Also, 
\begin{equation}
	\text{Var}(\|\boldsymbol y\|_2) = \mathbb{E}(\|\boldsymbol y\|^2_2) - \mathbb{E}^2(\|\boldsymbol y\|_2) = m\sigma^2 - m\sigma^2(1 - \Delta)^2
\end{equation}

Taking the extremal values $\mathbb{E}(\|\boldsymbol y\|_) = \sqrt{m}\sigma,\; \text{Var}(\|\boldsymbol y\|_1) = m\sigma^2$ the lemma follows from Chebyshev's inequality.
\end{proof}

\begin{customthm}{4}[Narrow gorge landscapes are unrecoverable (finite sampling noise)]
	If a family of $N$ qubit VQAs presents in its cost function a narrow gorge with any decay constant $b$, then for sufficiently large $N$ with high probability $\ge 1 - 1/{\normalfont \text{poly}}(N)$ it is not recoverable by BPDN if it is sampled on a quantum computer at $m \in O(\text{\normalfont poly}(N))$ randomly chosen points with at most polynomially many shots (measurements) per sample. 
\end{customthm}

\begin{proof}
	Suppose that, for each sampled point, $n_s \in O(\text{\normalfont poly}(N))$ shots are taken on the quantum computer. Then this leads to an error term estimate due to finite sampling noise which is:
	\begin{equation}
		\epsilon \propto \sqrt{\frac{1}{n_s}} \in \Omega(1/\text{\normalfont poly}(N))
	\end{equation}
	From Lemma \ref{lem:2norm}, if one chooses $t, m \in \text{\normalfont poly}(N)$, and one applies the narrow gorge condition (\ref{eq:var}): $\sigma \in O(b^{-N})$, then with high probability $\ge 1 - \frac{1}{\text{\normalfont poly}(N)}$ one has 
	\begin{equation}
		\|\boldsymbol{C}\|_2\in O(b^{-N}). 
	\end{equation}
	Therefore, the incorrect solution $\boldsymbol{\hat c}_{err} = \boldsymbol 0$ will be valid as for sufficiently large system sizes since $\|\boldsymbol{\Phi}\boldsymbol{\hat c}_{err} - \boldsymbol{C}\|_2 = \|\boldsymbol{C}\|_2 < \epsilon$.
\end{proof}

\section{Non-periodic cost functions}
\label{app:nonuniform}
In proving some of the results in this paper we have assumed that the cost functions are periodic, which is also implied by the term trigonometric polynomial. However, in some applications it may be necessary to consider non-periodic cost functions, for instance in QAOA where the problem Hamiltonian has coefficients which are not integer multiples of each other. In that case, Theorem \ref{thm:1} dictates that the frequencies, while still discrete, will not be uniformly separated, thus leading to a non-periodic function. In fact, under the assumption that the frequencies composing the function are discrete, these two conditions (non-periodicity and non-uniform frequency spectrum) can be seen to be entirely equivalent.

Using a non-uniform frequency spectrum leads to some modifications of the results. From Fourier analysis, it can be seen that the Nyquist sampling rate is still determined by the largest frequency present, however the minimal sampling range is now dictated by the size of the smallest interval between any two frequencies. A large range will thus be needed to resolve two modes that are close in Fourier space.
Where a complete sampling is available, FFT should be replaced with one of its various non-uniform implementations, which retains similar efficiency \cite{boyd1992fast, dutt1993fast, dutt1993fast, ruiz2018nonuniform}.
For sparse function reconstruction, Basis Pursuit and BPDN will largely remain valid as both techniques are not exclusive to Fourier basis but may be applied to arbitrary bases \cite{chen2001atomic}, however the relation between required number of sample points and sparsity (Equation \ref{eq:sampling_req}) may now no longer hold as it is specific to the uniform frequencies case.

Finally, the results on narrow gorges would hold by replacing $\boldsymbol{C}$ with $\boldsymbol{C} - \<\boldsymbol{C}_{all}\>$ in the definition and theorems, where now the vector of all samples $\boldsymbol{C}_{all}$ covers the minimum range as described above.

\end{document}